\documentclass{ecai} 

\usepackage{latexsym}
\usepackage{amssymb}
\usepackage{amsmath}
\usepackage{amsthm}
\usepackage{booktabs}
\usepackage{enumitem}
\usepackage{graphicx}
\usepackage{color}

\usepackage{algorithm}

\newtheorem{theorem}{Theorem}

\newtheorem{corollary}{Corollary}

\newtheorem{lemma}[theorem]{Lemma}
\newtheorem{claim}{Claim}[theorem]

\newtheorem{question}{Question}

\newcommand{\cA}{\mathcal{A}}
\newcommand{\cG}{\mathcal{G}}
\newcommand{\cP}{\mathcal{P}}
\newcommand{\cS}{\mathcal{S}}
\newcommand{\mms}{\mathsf{mms}}
\newcommand{\pmms}{\mathsf{pmms}}

\newcommand{\BibTeX}{B\kern-.05em{\sc i\kern-.025em b}\kern-.08em\TeX}

\begin{document}


\begin{frontmatter}


\paperid{123} 


\title{On Approximate MMS Allocations \\on Restricted Graph Classes}


\author[A]{\fnms{Václav}~\snm{Blažej}\thanks{Supported by the European Research Council (ERC) under the European Union’s Horizon 2020 research and innovation programme grant agreement number 948057.}}
\author[B]{\fnms{Michał}~\snm{Dębski}}
\author[B]{\fnms{Zbigniew}~\snm{Lonc}\thanks{Corresponding Author. Email: zbigniew.lonc@pw.edu.pl}} 
\author[C]{\fnms{Marta}~\snm{Piecyk}}
\author[A,B]{\fnms{Paweł}~\snm{Rzążewski}\footnotemark[*]}

\address[A]{University of Warsaw}
\address[B]{Warsaw University of Technology}
\address[C]{CISPA Helmholtz Center for Information Security}


\begin{abstract}
We study the problem of fair division of a set of indivisible goods with connectivity constraints. Specifically, we assume that the goods are represented as vertices of a connected graph, and sets of goods allocated to the agents are connected subgraphs of this graph. We focus on the widely-studied maximin share criterion of fairness. It has been shown that an allocation satisfying this criterion may not exist even without connectivity constraints, i.e., if the graph of goods is complete. In view of this, it is natural to seek approximate allocations that guarantee each agent a connected bundle of goods with value at least a constant fraction of the maximin share value to the agent. It is known that for some classes of graphs, such as complete graphs, cycles, and $d$-claw-free graphs for any fixed $d$, such approximate allocations indeed exist. However, it is an open problem whether they exist for the class of all graphs.

In this paper, we continue the systematic study of the existence of approximate allocations on restricted graph classes. In particular, we show that such allocations exist for several well-studied classes, including block graphs, cacti, complete multipartite graphs, and split graphs.
\end{abstract}

\end{frontmatter}


\section{Introduction}\label{intro}

Fair division of goods is a fundamental problem in social choice theory, computer science, and economics. 
The objective is to allocate goods to agents in a way that satisfies a specified fairness criterion. 
The study of this problem dates back to the seminal work of Steinhaus~\cite{s:48}. 
Extensive research has been conducted on the fair division of divisible goods (Brams and Taylor~\cite{bt:96}, Procaccia~\cite{p:16}). 
In this case, the problem is known as the \emph{cake-cutting problem}.

Another variant of the problem, that we deal with in this paper, is fair division of a set of indivisible goods. 
In this case, the problem involves a set of items and a set of agents, each with her own utility function for the goods, which assigns some values to all subsets of goods, called bundles. 
It is commonly assumed, as we do in this paper, that utility functions are additive, i.e., the value of each bundle is equal to the sum of the values of the individual items in this bundle. 

It has been shown that for divisible goods, allocations satisfying natural criteria such as envy-freeness or proportionality always exist (Brams and Taylor~\cite{bt:96}). 
However, when the goods are indivisible, it is easy to give examples where no allocation satisfies these criteria. 
Therefore, there is a need to define fairness criteria that are more realistic and achievable when allocating indivisible goods.

In recent years, one of the most extensively studied fairness criteria for the allocation of indivisible goods has been the \emph{maximin share} criterion, introduced by Budish~\cite{b:11}. 
The maximin share of an agent is the maximum value the agent can ensure to obtain by dividing the set of goods into $n$ parts, where $n$ is the number of agents, and receiving a part with the minimum value. 
An allocation satisfies the maximin share criterion if every agent receives a set of goods of value at least equal to her maximin share. 
Such allocations are referred to as \emph{maximin allocations}, or simply \emph{mms-allocations}. 

For several years it has been an open problem whether mms-allocations always exist. This question was answered in the negative (Procaccia and Wang~\cite{pw:14}, Feige et al.~\cite{fst:22}), even for as few as three agents. In light of this, it is natural to look for \emph{approximate mms-allocations}, which guarantee that each agent receives a bundle of value at least some positive fraction of her maximin share. Procaccia and Wang~\cite{pw:14} proved the existence of an allocation ensuring every agent a bundle of value at least $\frac{2}{3}$ of her maximin share. 
This approximation ratio has since been improved in several subsequent papers (Amanatidis et al.~\cite{amns:17}, Barman and Krishnamurthy~\cite{bm:20}, Garg et al.~\cite{gmt:19}, Ghodsi et al.,~\cite{ghssy:18}, Garg and Taki~\cite{gt:19}, Akrami et al.~\cite{agst:23}). 
The best known result of this kind is given by the following recent theorem.

\begin{theorem}[Akrami and Garg~\cite{ag:24}]\label{thm:completegraph34}
For any set of goods and any set of agents, there always exists an allocation in which each agent gets at least $\frac{3}{4}+\frac{3}{3836}$ of her maximin share.
\end{theorem}

In this paper, we study the problem of mms-allocations under some connectivity constraints imposed on the set of goods. We follow the framework proposed by Bouveret et al.~\cite{bceip:17}. 
In this framework, the items to be distributed among agents are represented as vertices of a connected graph, referred to as the \emph{graph of goods}. 
An allocation satisfies a fairness criterion if the values of the bundles received by agents are not only of sufficiently high value to them but also induce connected subgraphs in the graph of goods. 
This model captures a variety of applications. 
One of them is the problem of consolidating land plots (King and Burton~\cite{kb:82}). 
In this scenario, highly fragmented plots need to be reallocated to farmers in such a way that the new plots they receive are close one to another.
 Other examples include the allocation of rooms in a building to research groups (Bouveret
et al.,~\cite{bceip:17}) or the distribution of airline connections among carriers (Lonc~\cite{Lonc:23}). 
 The original problem of fair allocation of indivisible goods becomes a special case in this framework when the graph of goods is complete.
 
 The problem of the existence of mms-allocations and approximate mms-allocations has been investigated for several classes of graphs. 
Bouveret et al.~\cite{bceip:17} proved that mms-allocations always exist when the graph of goods is a tree. 
Moreover, they provided an example demonstrating that if the graph of goods is a cycle, then for certain utility functions of the agents, an mms-allocation may not exist. 

Lonc and Truszczynski~\cite{lt:20} thoroughly studied the case where the graph of goods is a cycle. 
They proved that in this case, for any set of agents, an allocation guaranteeing the agents bundles of value at least $\frac{\sqrt{5}-1}{2}\approx 0.62$ of their maximin shares always exists. 

Lonc~\cite{Lonc:23} considered the case where the graph of goods is $d$-claw-free, i.e., it does not contain a star with $d$ edges as an induced subgraph. 
He proved that for $d\geq 3$, if the graph of goods is $d$-claw-free,  then there exists an allocation ensuring every agent a bundle of value at least $\frac{1}{d-1}$ of her maximin share. 
It is worth noting that the class of $d$-claw-free graphs contains the class of graphs with a maximum degree bounded by $d-1$. 

In view of the results mentioned above, it is natural to ask the following question.

\begin{question}\label{quest1}
Is there a constant $\alpha>0$ such that for any connected graph of goods and any set of agents with arbitrary utility functions, there is an allocation that guarantees every agent a bundle of value at least $\alpha$ fraction of her maximin share?
\end{question}
To the best of our knowledge, this question remains wide open. In this paper, we provide a positive answer to Question~\ref{quest1} for certain restricted classes of graphs. 

\subsection{Our contribution}\label{contribution}
We continue the line of research focused on finding approximate allocations that guarantee agents a positive fraction of their maximin shares for several important classes of graphs of goods widely studied in graph theory. 

The first class we consider is the class of block graphs, which can be viewed as a common generalization of trees and complete graphs: here, complete graphs are arranged in a tree-like structure. 
Since a complete graph is a block graph, we know that we cannot always guarantee that each agent gets her maximin share. On the other hand, thanks to the tree-like structure, we may hope that the techniques developed for trees by Bouveret et al.~\cite{bceip:17} may be generalized to this setting. 
Actually, we consider the class of \emph{block-cactus graphs}, which encompasses both block graph and \emph{cacti} -- another well-studied class,  where cycles are arranged in a tree-like structure.

For this class of graphs, we show the following result.

\begin{theorem}\label{thm:blockcactusintro}
For every connected block-cactus graph and for any set of agents, there always exists an allocation in which each agent gets at least $\frac{1}{2}$ of her maximin share.
\end{theorem}

It is worth noting that the class of block-cactus graphs contains graphs with arbitrarily large stars as induced subgraphs, so our results are not covered by  previously known results. 

Next, we consider the class of complete multipartite graphs. 
This class can be seen as another generalization of complete graphs, where each vertex is replaced by a set of independent vertices. 
Notably, this class contains the important class of complete bipartite graphs. In particular, these graphs contain large induced stars.

\begin{theorem}\label{thm:multipartiteintro}
For every complete multipartite graph with at least two parts, and any set of agents, there always exists an allocation in which each agent gets at least $\frac{1}{4}$ of her maximin share.
\end{theorem}

The last class of graphs we study is the class of split graphs, i.e., graphs whose vertex sets can be partitioned into two subsets: one inducing a complete graph and the other being edgeless. 
Note that this class contains all complete graphs and graphs in this class may contain induced stars of arbitrarily large size. 
Our result for this class is somewhat weaker, as we show that for any set of agents of a fixed number of types (agents are of the same type if they have the
same utility function), there exists an allocation that ensures each agent a constant fraction of her maximin share.

\begin{theorem}\label{thm:splitintro}
For every connected split graph and every set of agents of at most $2^k$ types, there always exists an allocation in which each agent gets at least $\frac{3}{7\cdot 2^k - 3}$ of her maximin share.
\end{theorem}

The classes of block-cactus graphs, complete multipartite graphs, and split graphs all contain the class of complete graphs.
Feige, Sapir, and Tauber~\cite{fst:22} constructed an example demonstrating that if the graph of goods is complete and has at least 9 vertices, then for certain utility functions of the agents, no allocation can guarantee the agents more than $\frac{39}{40}$ of their maximin shares. 
This result establishes an upper bound of $\frac{39}{40}$ on the approximation ratio for approximate mms-allocations within the classes of graphs considered in this paper.

All the classes of graphs considered in this paper are \emph{hereditary}, i.e., they are closed under vertex deletion. 
This property  allows us to use a  lemma of Lonc~\cite{Lonc:23}, which enables us to reduce the problem of existence of approximate allocations for arbitrary agents to the case where agents' utility functions are bounded by a constant fraction of their maximin share. 
To apply this lemma, we need to admit graphs of goods to be disconnected, which requires  extending the definitions of maximin share and mms-allocation. 
Nevertheless, disconnected graphs and the related concepts play an auxiliary role in this paper. 

\subsection{Related work}
The fair division of a graph into connected bundles received a significant attention in recent years (see the survey papers by Suksompong~\cite{s:21} and Biswas et al.~\cite{bpsv:23}). 
We have already mentioned the contributions of Bouveret et al.~\cite{bceip:17}, who initiated research in this area, the paper of Lonc and Truszczynski~\cite{lt:20}, who focused on the case where the graph of goods is a cycle, and Lonc~\cite{Lonc:23}, who examined the case where the graph of goods is $d$-claw-free.

Bei et al.~\cite{bils:22} introduced the concept of the {\it price of connectivity} for a graph $G$, defined as the worst-case ratio between the maximin share calculated with the constraint that all bundles must be connected in $G$ and the maximin share without this constraint. 
One of their results most relevant to the problems considered in this paper (Theorem 3.14 in Bei et al.~\cite{bils:22}) shows that for any graph of goods $G$, there exists an allocation guaranteeing each agent a bundle of value at least $\frac{1}{m-n+1}$ of her maximin share, where $n$ is the number of agents, $m$ is the number of vertices in $G$, and $m\geq n$. 
However, this guarantee becomes quite weak as the size of the graph increases. 
In contrast, the allocations demonstrated in this paper guarantee each agent a fraction of her maximin share that is independent of the size of the graph of goods. 

Bouveret et al.~\cite{bcl:19} and Xiao et al.~\cite{xqh:23} investigated the problem of connected fair allocation of indivisible chores (i.e., items that generate disutility for the agents). 
Greco and Scarcello~\cite{gs:24} examined maxileximin allocations under connectivity constraints, which aim to minimize agents' dissatisfaction while maximizing social welfare. 

Bilò et al.~\cite{Bilo:19} and Bei et al.~\cite{bils:22} established several results on relaxations of envy-free allocations with connected bundles. Suksompong~\cite{s:19} presented various approximation results for fairness criteria such as envy-freeness, proportionality, and equitability when the graph of goods is a path. 
Igarashi and Peters~\cite{ip:19} explored the allocation of connected bundles that are Pareto optimal.
Greco and Scarcello~\cite{gs:19} and Deligkas et al.~\cite{deli:21} focused on computational complexity issues related to fair allocation on graphs. 
Finally, Hummel and Igarashi~\cite{hi:24} addressed local fairness in graph-based division and focused on pairwise maximin share allocations.

\section{Preliminaries}
Let $N=[n]$ be a set of agents and let $V$ be a set of goods.
We assume that $V$ is a vertex set of an undirected graph that we call a {\it graph of goods}.
We do not assume that the graph $G$ is connected.
We say that a subset $X\subseteq V$ is a {\it bundle} if $X$ induces a connected subgraph of $G$.
In particular, the empty set is a bundle. 

For every agent $i\in N$ we define a {\it utility function} $u_i$ which assigns nonnegative reals to the goods in $V$. 
We extend the functions $u_i$ to sets of goods by assuming additivity.
So, for an arbitrary set $X\subseteq V$ we define $u_i(X):=\sum_{v\in X}u_i(v)$.
By convention, we use $u_i(\emptyset)=0$. We say that agents are of the {\it same type} if they have the same utility function. 

We call a tuple  $(P_1,\ldots,P_n)$ of bundles a {\it $(G,n)$-packing}, if the sets $P_1,\ldots,P_n$ are pairwise disjoint.
If, in addition, $\bigcup_{i=1}^nP_i=V$, then the $(G,n)$-packing is called a {\it $(G,n)$-partition}.
Let $\cP(G,n)$ (respectively, $\cS(G,n)$) denote the set of all $(G,n)$-partitions (respectively, $(G,n)$-packings) of the set $V$.

When we say that a $(G,n)$-packing (or a $(G,n)$-partition as a special case) $(P_1,\ldots,P_n)$ is an \emph{allocation},
we mean that for each $i \in [n]$, the bundle $P_i$ is assigned to the agent $i$.

Then, for an agent with utility function $u$ we define the {\it maximin share} as 
\[
\mms^{(n)}(G,u) = \max_{(P_1,\ldots,P_n)  \in \cP(G,n)} \min_{j\in [n]} u(P_j).
\]
Any $(G,n)$-partition of the set $V$ for which the maximum in the definition above is attained is called an {\it mms-partition} for the agent with utility function $u$. Note that even though mms-partition is defined as an ordered tuple, permuting the bundles arbitrarily also gives an mms-partition.

For technical reasons, it will also be convenient to admit disconnected graphs of goods.
For such graphs the maximin share may be undefined.
Indeed, consider a graph $G$ with more than $n$ components.
For such a graph no $(G,n)$-partition exists, so the maximin share is undefined.
To get around of this problem we replace in the definition of the maximin share the set $\cP(G,n)$ of all $(G,n)$-partitions of $V$ by the set $\cS(G,n)$ of all $(G,n)$-packings of $V$. So, by a {\it packing maximin share} for an agent with utility function $u$ we mean the value
\[
\pmms^{(n)}(G,u) = \max_{ (P_1,\ldots,P_n) \in \cS(G,n)} \min_{j \in [n]} u(P_j).
\]
This parameter is well-defined for every (not necessarily connected) graph.
We call a $(G,n)$-packing for which the maximum in the definition above is attained an {\it mms-packing} for the agent with utility function $u$. Clearly, $\pmms^{(n)}(G,u)\geq \mms^{(n)}(G,u)$ for every graph $G$ for which the maximin share is defined because ${\cal P}(G)\subseteq {\cal S}(G)$. Moreover, since the value of every bundle in a $(G,n)$-packing $(P_1,\ldots,P_n)$ for any agent $i$ is at least $\pmms^{(n)}(G,u_i)$, we have
\begin{align}
\begin{split}
u_i(V)\geq &\sum_{j=1}^nu_i(P_j) \\ \geq & \ n\cdot \pmms^{(n)}(G,u_i)\geq n\cdot \mms^{(n)}(G,u_i).
\label{eq4}
\end{split}
\end{align}
It is known (see Lemma~\ref{prop_old}) that $\pmms^{(n)}(G,u)=\mms^{(n)}(G,u)$ for every connected graph $G$. For disconnected graphs the values of the maximin share and the packing maximin share may differ. To see this, consider a graph $G$ with three vertices $x,y,z$, and a single edge connecting $x$ and $y$. Suppose there are two agents with the same  utility function $u$, where $u(x)=u(y)=2$ and $u(z)=1$. In this case, we have $\mms^{(2)}(G,u)=1$, while $\pmms^{(2)}(G,u)=2$. 

We say that an allocation $(A_1,\ldots,A_n) \in \cP(G,n)$ (respectively, $(A_1,\ldots,A_n) \in \cS(G,n)$) is an {\it mms-allocation} (respectively, a {\it pmms-allocation}) if 
\begin{align*}
u_i(A_i) &\geq \mms^{(n)}(G,u_i)\\
{\rm (respectively,}\ u_i(A_i) &\geq \pmms^{(n)}(G,u_i){\rm )}
\end{align*}
for every agent $i\in N$. We emphasize that in an mms-allocation, every element of $V$ is assigned to some agent, while this is not necessarily the case for pmms-allocations.

Similarly, for any $\alpha>0$, an allocation $(A_1,\ldots,A_n) \in \cP(G,n)$ (respectively, $(A_1,\ldots,A_n) \in \cS(G,n)$) is an {\it $\alpha$-mms-allocation} (respectively, a {\it $\alpha$-pmms-allocation}) if 
\begin{align*}
u_i(A_i)&\geq  \alpha\cdot \mms^{(n)}(G,u_i) \\
{\rm (respectively,}\ u_i(A_i)&\geq   \alpha\cdot \pmms^{(n)}(G,u_i){\rm )}
\end{align*}
for every agent $i\in [n]$. 
Note that 1-mms-allocations (respectively, 1-pmms-allocations) are exactly mms-allocations (respectively, pmms-allocations).

An agent $i\in N$ with utility function $u_i$ is {\it $\alpha$-mms-bounded} for the graph of goods $G$, if 
\[
u_i(v)<\alpha\cdot \mms^{(n)}(G,u_i)
\]
for every good $v$ in $G$. Note that in this case we must have $\mms^{(n)}(G,u_i) > 0$.

If the graph of goods $G$ and the number $n$ of agents are clear from the context, then we will write $\mms_i$ instead of $\mms^{(n)}(G,u_i)$ and $\pmms_i$ instead of $\pmms^{(n)}(G,u_i)$ for each agent $i\in [n]$. 

In this paper we use a standard graph-theoretic terminology and notation.
In particular, by $V(G)$ we denote the set of vertices of a graph $G$.
For a set $X \subseteq V(G)$, by $G-X$ we denote the graph obtained from $G$ by removing vertices of $X$ and all their incident edges.
We say that a subgraph $H$ of a graph $G$ is {\it induced by a set of vertices $X\subseteq V(G)$} if $H = G - (V(G) \setminus X)$.
A class of graphs $\cG$ is {\it hereditary} if every induced subgraph of a graph in $\cG$ is in $\cG$ too. 
In other words, $\cG$ is closed under vertex deletion.

A vertex $v$ in a connected graph $G$ is called a {\it cut-vertex} if the graph obtained from $G$ by removing $v$ is no longer connected.
A connected graph is {\it biconnected} if it does not have any cut-vertex.
A maximal biconnected subgraph of a graph $G$ is called a {\it block} of $G$.
It is well-known (see Bondy and Murty~\cite{bm:08}) that any two distinct blocks in a graph $G$ share at most one common vertex, and if such a vertex exists, then it is a cut-vertex of $G$.
For a connected graph $G$, we define a bipartite graph $B(G)$,
where one bipartition class consists of the blocks of $G$, and the other consists of the cut-vertices of $G$.
A block $B$ in $G$ and a cut-vertex $c$ in $G$ are adjacent in $B(G)$ if $c$ belongs to the vertex set of $B$.
It is easy to verify that the graph $B(G)$ is a tree.
Every block that is a leaf in this tree is called a {\it terminal block}. 

A graph in which every block is a complete graph is called a {\it block graph}, and a graph in which every block is a cycle or an edge is called a {\it cactus}. Note that both these classes contain all trees.

A graph $G$ is {\it complete multipartite} if its vertex set can be partitioned into 
subsets called \emph{parts} such that two vertices are adjacent in $G$ if and only if they belong to distinct parts. 
Note that each part is {\it independent}, i.e., there is no edge in $G$ contained in it. 
We observe that a complete multipartite graph is connected unless it has just one part and more than one vertex. 
Furthermore, the class of complete multipartite graphs is, obviously, hereditary.

A graph $G=(V,E)$ is called a \emph{split graph} if there exists a partition of its vertices into two sets $K,I$, such that $I$ is an independent set and $K$ induces a complete graph.
Again, the class of split graphs is hereditary.

We conclude this section with two simple observations that will be useful later.
The following lemma appearing in Lonc~\cite{Lonc:23} explains the relationship between mms-allocations and pmms-allocations when the graph of goods is connected.
\begin{lemma}\label{prop_old}
Let $G$ be a connected graph of goods.

\noindent
(i) For any utility function $u$ defined on the set of vertices of $G$ and a positive integer $n$, $\pmms^{(n)}(G,u) = \mms^{(n)}(G,u)$.\\
(ii) For any $\alpha > 0$ and any collection of $n\geq 1$ agents, an $\alpha$-mms-allocation exists if and only if an $\alpha$-pmms-allocation exists. 
\end{lemma}

Here is another statement that will be useful in this paper. Its weaker version appeared in Lonc~\cite{Lonc:23} as Lemma 2.
\begin{lemma}\label{obs:heavy}
Let $\cG$ be a hereditary class of graphs, $G\in \cG$, and $\alpha>0$.
If there exists an $\alpha$-mms-allocation for any connected induced subgraph of $G$ and any set of $\alpha$-mms-bounded agents (respectively, any set of $\alpha$-mms-bounded agents of at most $t$ types),
then there exists an $\alpha$-pmms-allocation for $G$ and an arbitrary set of agents (respectively, an arbitrary set of agents of at most $t$ types).
\end{lemma}
\begin{proof} 
In the proof we will apply the following statement appearing in Lonc~\cite{Lonc:23} as Proposition 1 (iii):

\smallskip
\noindent
{\it Let $G$ be a connected graph of goods. For any utility function $u$ defined on $V(G)$ and positive integers $m,n$ such that $m\leq n$, 
\begin{equation}\label{eq_lonc}
\pmms^{(m)}(G,u) \geq \pmms^{(n)}(G,u).
\end{equation}}

Consider an arbitrary graph of goods $G\in{\cal G}$ and an arbitrary set $N=\{ 1,\ldots,n\}$ of agents. Let, for each agent $i\in N$, $u_i$ be the utility function of $i$ and let ${\cal P}_i$ be an mms-packing for this agent. To simplify notation we define $\pmms_i:=\pmms^{(n)}(G,u_{i})$ for every agent $i\in N$. We need to construct an allocation assigning each agent $i$ a bundle of value at least $\alpha\cdot \pmms_i$. 

We define  $X:=\{  x_1,\ldots,x_\ell\}$ to be a maximal with respect to inclusion (and possibly empty) set of pairwise different vertices of $G$ such that every vertex $x_j$ is of value at least $\alpha\cdot \pmms_{i_j}$ to some agent $i_j$ and the  agents $i_1,\ldots,i_\ell$ are pairwise different. Every agent $i_j$ receives the vertex $x_j$ in our allocation.

It remains to show that there is an allocation assigning to each agent $i\in N'=N-\{ i_1,\ldots,i_\ell\}$ a bundle inducing a connected subgraph of $G-X$ which is of value at least $\alpha\cdot \pmms_i$ to this agent.
By the maximality of $X$, for every vertex $v\in V(G-X)$ and every agent $i\in N'$, 
\begin{equation}\label{ineq2}
u_i(v)<\alpha\cdot \pmms_i.
\end{equation}
Since $|X|=\ell$, for every agent $i\in N'$ there are at least $n-\ell$ bundles in ${\cal P}_i$ which do not intersect $X$. 

Let $G_1,\ldots,G_s$ be the components of $G-X$. 
Let ${\cal P}_i^j$ be the set of bundles in ${\cal P}_i$ that are contained in $V(G_j)$ and let $f(i,j)$ be the number of such bundles. Clearly, ${\cal P}_i^j$ is a $(G_j,f(i,j))$-packing. Moreover, $\sum_{j=1}^sf(i,j)\geq n-\ell$ 
for each agent $i\in N'$. We observe that, by Lemma 5 (i) and the fact that the value of each of $f(i,j)$ bundles in ${\cal P}_i^j$ is at least $\pmms_i$ for agent $i$, we have 
\begin{equation}\label{ineq1}
\mms^{(f(i,j))}(G_j,u_i)=\pmms^{(f(i,j))}(G_j,u_i)\geq \pmms_i
\end{equation}
for every agent $i\in N'$ and every $j\in\{ 1,\ldots,s\}$.

\medskip
\noindent
{\bf Claim.} Let $I_j\subseteq N'$ be a set of agents such that $f(i,j)\geq |I_j|$ for each agent $i\in I_j$.  Then, there is an $\alpha$-mms allocation for the graph of goods $G_j$ and the set of agents $I_j$. 

\begin{algorithm}[!ht]
\begin{tabbing}
1\quad\=\quad\=$T:=N'$;\\
2\>\>{\bf for} $j=1,2,\ldots,s$ {\bf do}\\
3\>\>\quad\=sort the agents $i\in T$ non increasingly according to \\
\>\>\>\quad\=the key $f(i,j)$: $i^j_1,i^j_2,\ldots,i_{|T|}^j$;\\
4\>\>\>$k_j:=$ the largest $p$ such that $f(i_p^j,j)\geq p$;\\
5\>\>\>$I_j:=\{ i^j_1,i^j_2,\ldots,i^j_{k_j}\}$;\\
6\>\>\>agents in $I_j$ distribute the vertices of  $V(G_j)$ among \\
\>\>\>\>themselves according to the allocation $A(I_j,G_j)$;\\
7\>\>\>$T:=T-I_j$;\\
8\>\>\>{\bf if} $T=\emptyset$ {\bf then}\\
9\>\>\>\>{\bf return}
\end{tabbing}
\caption{${\it allocate}(N',G-X,c)$}\label{alg1}
\end{algorithm}
\noindent
{\it Proof of the Claim.} We observe that the agents of $I_j$ are $\alpha$-mms bounded with respect to the graph of goods $G_j$ and  the utility functions restricted to $V(G_j)$. Indeed, applying in turn the inequalities (\ref{ineq2}), (\ref{ineq1}) and (\ref{eq_lonc}), for every vertex $v\in V(G_j)$ we get
\begin{align*}
u_i(v) &< \alpha\cdot \pmms_i\leq \alpha\cdot \mms^{(f(i,j))}(G_j,u_i)\\
&\leq \alpha\cdot \mms^{(|I_j|)}(G_j,u_i).
\end{align*}
Moreover, if the agents in $N$ are of at most $t$ types, then the agents in $I_j$ are also of at most $t$ types, as $I_j\subseteq N$. 
Since $G_j$ is a connected induced subgraph of $G$, by our theorem assumption, there is an $\alpha$-mms allocation for the graph of goods $G_j$ and the set of agents $I_j$. This completes the proof of the Claim.\hfill

\medskip
We denote by $A(I_j,G_j)$ the allocation whose existence is guaranteed by the Claim. 


We shall prove that the algorithm {\it allocate} shown in the figure Algorithm~\ref{alg1} produces the required $\alpha$-pmms allocation of the vertices of $G-X$ to the agents of  $N'$. In the $j$th pass of the loop of this algorithm bundles of vertices of the component $G_j$ of $G-X$ are distributed to agents. First, we sort (line 3) the agents $i$ which have not got their bundles yet  non increasingly according to the number $f(i,j)$ of bundles in their mms-packings ${\cal P}_i$ which are contained in $G_j$. We denote the $\ell$th agent in this sorting by $i^j_\ell$. Next, we define the set $I_j$ which consists of $k_j$ initial agents in the sorting  (line 5), where $k_j$ is the largest $p$ such that $f(i^j_p,j)\geq p$ (line 4) . We observe that, by the definition of $k_j$, $f(i,j)\geq k_j=|I_j|$ for every agent $i\in I_j$. Thus, it follows from the Claim that there is an allocation $A(I_j,G_j)$ assigning to each agent $i\in I_j$ a bundle of value at least 
\begin{align*}
\alpha\cdot \mms^{(k_j)}(G_j,u_i)&=\alpha\cdot \pmms^{(k_j)}(G_j,u_i)
\\ &\geq \alpha\cdot \pmms^{(f(i,j))}(G_j,u_i)\geq \alpha\cdot \pmms_i.
\end{align*}
We applied here, in turn, Lemma 5 (i), and the inequalities (\ref{eq_lonc}) and (\ref{ineq1}). We distribute the vertices of $G_j$ to agents of $I_j$ according to the allocation $A(I_j,G_j)$ (line 6) and the agents of $I_j$ quit the game (line 7). 

It remains to show that when the algorithm stops, the set $T$ is empty, i.e., all agents received their bundles. 
Suppose otherwise. Let $i$ be an agent who is still in the set $T$ when the algorithm stops. 
Let $T_j$ be the set $T$ at start of the $j$th pass of the loop of our algorithm. By the definition of $k_j$, for all agents $t\in T_j$, which are not included in the set $I_j$ (so do not receive a bundle) in the $j$th pass of the loop, we have $f(t,j)<k_j+1$, i.e., $f(t,j)\leq k_j$. 
Since the agent $i$ was not allocated a bundle in any pass of the loop, we have $f(i,j)\leq k_j$ for all $j$'s.  
Clearly, $\sum_{j=1}^sk_j$ is the total number of agents who received their bundles when the algorithm stops. Hence, 
\[
n-\ell\leq\sum_{j=1}^sf(i,j)\leq\sum_{j=1}^sk_j<n-\ell, 
\]
a contradiction. Thus, all agents received their bundles when the algorithm stops. 
\end{proof}

\setcounter{theorem}{1}

\section{Block graphs and cactus graphs}
We say that $G$ is a \emph{block-cactus graph} if its every block is either a cycle or a complete graph.
Clearly this class is a common generalization of cacti and block graphs.
Furthermore, it is easy to verify that block-cactus graphs form a hereditary class.
We will prove the following theorem which is a reformulation of the result mentioned in Subsection~\ref{contribution} using the introduced terminology.

\begin{theorem}\label{thm:block-cacti}
For every connected block-cactus graph and for any set of agents there exists a $\frac{1}{2}$-mms-allocation. 
\end{theorem}
\begin{proof}
We observe that to prove the theorem it suffices to show the following statement: 

\medskip
\noindent
{\sl For every connected block-cactus graph and any set of $\frac{1}{2}$-mms-bounded agents, a $\frac{1}{2}$-mms-allocation exists.}

\medskip
\noindent
Indeed, if we proved it, then the theorem follows immediately by Lemmas~\ref{prop_old} and~\ref{obs:heavy}.

We will prove the statement above by induction on the number of vertices in the graph $G$.
As a base case, suppose that $G$ is biconnected; in particular, this covers the case that $G$ has one vertex.
Then, $G$ is either a cycle or a complete graph and the existence of the required mms-allocation is well known (see, e.g., Akrami and Garg~\cite{ag:24} and Lonc and Truszczynski~\cite{lt:20}).

Thus, suppose that $G$ is a connected block-cactus graph that has at least one cut-vertex.
Let $N=[n]$ be a set of $\frac{1}{2}$-mms-bounded agents.

We will show that for this graph and this set of agents, a $\frac{1}{2}$-mms-allocation exists. 
For every agent $i\in N$ we denote by $u_i$ her utility function and by $\cP_i=( P_1^i,\ldots,P_n^i )$ her mms-partition.
Recall that we defined $\mms_i:=\mms^{(n)}(G,u_i)$.

Let $B$ be the set of vertices of a terminal block of $G$ and let $v$ be the unique cut-vertex contained in $B$.
Let $B':=B \setminus \{ v\}$. We consider two cases.

\smallskip
\noindent
{\it Case 1: For each agent $i\in N$, it holds $u_i(B')< \mms_i$.} 

\smallskip
Let $G'$ be the graph obtained from $G$ by removing all vertices in $B'$.
For each agent $i$ we define a new utility function $u'_i$ on $V(G')$ as follows: $u'_i(v)=u_i(B)$ and $u'_i(w):=u_i(w)$ for every $w\not=v$. Clearly, the agents with these modified utility functions may not be $\frac{1}{2}$-mms-bounded because it can happen that $u_i'(v)>\frac{1}{2}\mms^{(n)}(G',u_i')$. 

However, by the induction hypothesis applied for every connected induced subgraph $H$ of $G'$, there exists a $\frac{1}{2}$-mms-allocation for the graph $H$ and any set of $\frac{1}{2}$-mms-bounded agents with utility functions defined on the set $V(H)$.
It now follows from Lemmas~\ref{obs:heavy} and~\ref{prop_old} that there exists a $\frac{1}{2}$-mms-allocation for the graph $G'$ and an arbitrary set of agents.
In particular, this holds for the set of agents $N$ with utility function $u'_i$ for each agent $i\in N$.
Let $\cA$ denote such an allocation.
In this allocation, every agent $i\in N$ receives a bundle of value at least $\frac{1}{2}\mms^{(n)}(G',u'_i)$. 

It follows from the inequality $u_i(B')< \mms_i$ that for every agent $i$, the set $B$ is contained in one of the bundles, say $P_1^i$, of the mms-partition $\cP_i$.
Hence, since $u_i'(P_1^i \setminus B')=u_i(P_1^i)$ for each $i$, the value of each bundle in the partition $( P_1^i \setminus B',P_2^i,\ldots,P_n^i)$ of $V(G')$ under $u_i'$ is at least $\mms_i$. Thus, 
\[
\mms^{(n)}(G',u'_i)\geq \mms_i,
\]
and the allocation obtained from $\cA$ by extending the (unique) bundle in $\cA$ containing $v$ with the set $B'$ is a $\frac{1}{2}$-mms-allocation for the graph $G$ and the set of agents $N$ with utility function $u_i$ for each agent $i\in N$. This completes the proof for Case 1.

\smallskip
\noindent
{\it Case 2: There is an agent $i' \in N$ such that $u_{i'}(B')\geq \mms_{i'}$.} 

\smallskip
Let $v_1,\ldots,v_m$, where $v_m=v$, be the consecutive vertices of a Hamiltonian path in the block whose vertex set is $B$.
Let $j_1$ be the least index $j<m$ such that, for some agent $i_1\in N$, the set $L_1:=\{ v_1,\ldots,v_{j}\}$ is of value at least $\frac{1}{2}\mms_{i_1}$.
This agent receives the set $L_1$ and quits the game. Note that $L_1$ and $i_1$ are well-defined, as we are in Case 2.

Next, we define $j_2$ to be the least index $j<m$ such that, for some agent $i_2\in N \setminus \{i_1\}$, the set $L_2:=\{ v_{j_1+1},\ldots,v_{j}\}$ is of value at least $\frac{1}{2}\mms_{i_2}$, if such index exists. The agent $i_2$ receives the set $L_2$ and quits the game.
We continue such distribution of sets of vertices to agents as long as possible.
Let $\ell \geq 1$ be the number of agents that left the game in this phase and let $j_\ell$ be the largest index $j<m$ defined on the way.
So, in the last step, some agent $i_\ell\in N \setminus \{ i_1,\ldots,i_{\ell-1}\}$, received the set $L_\ell:=\{ v_{j_{\ell-1}+1},\ldots,v_{j_\ell}\}$, left the game, and for all agents in $i\in N':=N \setminus \{ i_1,\ldots,i_\ell\}$, the value of the set $R:=\{ v_{ j_\ell+1},\ldots,v_{m-1}\}$ is less than $\frac{1}{2}\mms_i$. 

Let $L:=L_1\cup\ldots\cup L_\ell$ and define $G':=G-L$. 
Clearly, the graph $G'$ is a connected block-cactus graph.
Furthermore, $G'$ has fewer vertices than $G$, as $\emptyset \neq L_1 \subseteq L$.

Consider any agent $i\in N'$. We shall prove that
\begin{equation}\label{eq3}
\mms^{(n-\ell)}(G',u_i)\geq \mms_i.
\end{equation}
We observe that $u_i(L_q)<\mms_i$, for every $q\in [\ell]$.
Indeed, by the definition of $L_q=\{ v_{j_{q-1}+1},\ldots,v_{j_q}\}$, we have $u_i(L_q \setminus \{ v_{j_q}\})<\frac{1}{2}\mms_i$.
Since the agent $i$ is $\frac{1}{2}$-mms-bounded, $u_i(L_q)<\mms_i$, as claimed.
Hence, 
\begin{equation}\label{eq2}
u_i(L)<\ell\cdot \mms_i.
\end{equation}

By this inequality and the inequality $u_i(R)<\frac{1}{2}\mms_i$ observed earlier, we conclude that $u_i(B')=u_i(L\cup R)<(\ell+1)\cdot \mms_i$.
Hence, at most $\ell +1$ bundles of the mms-partition $\cP_i$ intersect $B'$.
Indeed, all such bundles (each of value at least $\mms_i$ which is in turn positive, as agent $i$ is $\frac{1}{2}$-mms-bounded), except possibly one, must be contained in $B'$.
It follows from the observation that every bundle intersecting but not contained in $B'$ must contain the cut-vertex $v$. 

Suppose first that the set $L\subseteq B'$ intersects at most $\ell$ bundles of the mms-partition $\cP_i$. Then, at least $n-\ell$ bundles of $\cP_i$ are contained in $G'$. Thus, 
\[
\mms^{(n-\ell)}(G',u_i)=\pmms^{(n-\ell)}(G',u_i)\geq \mms_i
\]
which completes the proof of the inequality (\ref{eq3}) in this case. 

Now, assume that $L$ intersects exactly $\ell+1$ bundles of the mms-partition ${\cal P}_i$. We can assume without loss of generality that the bundles $P_1^i,\ldots,P_{\ell}^i$ are contained in $B'$ and the bundle $P_{\ell+1}^i$ is not contained in $B'$, i.e. it contains $v$. Since 
\begin{align*}
u_i(L)+u_i(R)=u_i(B')&= \sum_{q=1}^\ell u_i(P_q^i)+u_i(P_{\ell+1}^i\cap B')\\
&\geq \ \ell\cdot \mms_i+u_i(P_{\ell+1}^i\cap B'), 
\end{align*}
by this inequality and the inequality (\ref{eq2}), we have 
\[
u_i(R)> u_i(P_{\ell+1}^i\cap B'). 
\]
We replace in the bundle $P_{\ell+1}^i$ the set $P_{\ell+1}^i\cap B'$ by the set $R$. This way we get a bundle contained in $G'$ of value larger than $\mms_i$ for the agent $i$. This bundle together with the bundles $P_{\ell+2}^i,\ldots,P_n^i$ form a partition of the vertex set of $G'$ such that the value of each of these bundles for the agent $i$ is at least $\mms_i$. Thus, $\mms^{(n-\ell)}(G',u_i)\geq \mms_i$. This completes the proof of the inequality (\ref{eq3}).

Given that the agents in $N$ are $\frac{1}{2}$-mms-bounded for the graph $G$ and using the inequality (\ref{eq3}), we have 
\[
u_i(w)<\frac{1}{2}\cdot \mms_i\leq \frac{1}{2}\cdot \mms^{(n-\ell)}(G',u_i),
\]
for every vertex $w \in V(G')$ and every agent $i\in N'$. Therefore, the agents in $N'$ are $\frac{1}{2}$-bounded for the graph $G'$. 

By the induction hypothesis applied for the graph $G'$ and the set $N'$ of $\frac{1}{2}$-bounded agents, there exists a $\frac{1}{2}$-mms-allocation, say $\cal A$, for this graph $G'$ and this set of agents. By inequality (\ref{eq3}), in this allocation, each agent receives a bundle of value at least $\frac{1}{2}\cdot \mms^{(n-\ell)}(G',u_i)\geq \frac{1}{2}\cdot \mms_i$. 

The allocation $\cA$ together with the sets $L_1,\ldots,L_\ell$ allocated to the agents of $N \setminus N'$ form a $\frac{1}{2}$-mms-allocation for the graph $G$ and the set of agents $N$. 
\end{proof}

\section{Complete multipartite graphs}
In this section we prove Theorem~\ref{thm:multipartiteintro}, which we restate here using the introduced notation.

\begin{theorem}
For every complete multipartite graph with at least two parts, and any set of agents, a $\frac{1}{4}$-mms-allocation exists.
\end{theorem}
\begin{proof}
Let $G$ be a complete multipartite graph with at least two parts.  
Since the class of multipartite graphs is hereditary and the graph $G$ is connected, it follows from Lemmas~\ref{obs:heavy} and~\ref{prop_old} that it suffices to prove the theorem for a set of $\frac{1}{4}$-bounded agents.
Let $N:=[n]$ be a set of such agents and, for each $i\in N$, let $u_i$ be the utility function for the agent $i$.
Recall that $\mms_i:=\mms^{(n)}(G,u_i)$. 

Since each vertex of $G$ is of value less than $\frac{1}{4}\mms_i$ for any agent $i\in N$, each bundle of an mms-partition of any agent has at least $5$ vertices.
Thus, the total number of vertices in bundles of this mms-partition is at least $5n$.
Consequently, $|V(G)|\geq 5n$.

Let $X_1,\ldots,X_k$, for some $k \geq 2$, be the parts of $V(G)$, i.e., two vertices of $G$ are adjacent if and only if they belong to distinct parts.
We can assume without loss of generality that $|X_1|\leq\ldots\leq|X_k|$.
Let $\ell$ be the smallest index such that $|X_1|+\ldots+|X_\ell|\geq n$.
We observe that $\ell<k$.
Indeed, suppose $\ell=k$.
Then, $|X_1|+\ldots+|X_{k-1}|<n$.
Since $|V(G)|\geq 5n$, we have $|X_k|>4n$.
Consider an mms-partition $\cP$, for any agent in $N$.
Let $t$ be the number of bundles in this partition intersecting $X_k$.
Since $X_k$ is an independent set of vertices in $G$ and the bundles of $\cP$ are connected, each of these $t$ bundles must contain a different vertex of $X_1\cup\ldots\cup X_{k-1}$.
Thus, the set $X_1\cup\ldots\cup X_{k-1}$ contains at least $t$ vertices of the bundles of $\cP$ intersecting $X_k$ and at least $5(n-t)$ vertices of the bundles not intersecting $X_k$. Hence,
\[
|X_1|+\ldots+ |X_{k-1}|\geq t+5(n-t)=5n-4t\geq n,
\]
a contradiction, so $\ell<k$ as claimed.

We will show now that $|X_{\ell+1}|+\ldots+|X_k|\geq n$. This statement is obviously true if $|X_{\ell+1}|\geq n$, so assume that $|X_{\ell+1}| <n$. Then, $|X_\ell|\leq|X_{\ell+1}|<n$. By the definition of $\ell$, we have $|X_1|+\ldots+|X_{\ell-1}|<n$, so $|X_1|+\ldots+|X_{\ell}|<2n$. Since $|V(G)|\geq 5n$,
\begin{align*}
|X_{\ell+1}|+\ldots+|X_k|&= |V(G)|-(|X_1|+\ldots+|X_\ell|) \\ &\geq 5n-2n\geq n.
\end{align*}

We define $V_1:=X_1\cup\ldots\cup X_{\ell}$ and $V_2:=X_{\ell+1}\cup\ldots\cup X_k$. We have just proved that $|V_1|\geq n$ and $|V_2|\geq n$. By the definition of $G$, every vertex of $V_1$ is joined by an edge with every vertex of $V_2$. 
We partition the set of agents into two subsets $N_1$ and $N_2$. The set $N_1$ consists of all agents $i$ such that $u_i(V_1)\geq u_i(V_2)$ and $N_2:=N \setminus N_1$. By the inequality (\ref{eq4}), for each $i\in N$,  it holds that
\[
u_i(V_1)+u_i(V_2)=u_i(V(G))\geq n\cdot \mms_i.
\]
Hence, 
\begin{equation}\label{multi0}
\begin{split}
u_i(V_1)\geq \frac{n}{2} \cdot \mms_i & \quad \text{ if } i\in N_1,\\
u_i(V_2)\geq \frac{n}{2} \cdot \mms_i & \quad \text{ if } i\in N_2.
\end{split}
\end{equation} 

We now proceed to define the required allocation of the vertices of $G$ to the agents.
We assume, without loss of generality, that $N_1=\{ 1,\ldots,n_1\}$, for some $n_1\leq n$.
Let $o_v=(v_1,\ldots,v_{m_1})$, where $m_1:=|V_1|$, be an arbitrary ordering of the vertices of $V_1$.
The vertices of $V_1$ are assigned to agents in the order $o_v$ according to the following iterative procedure.

Suppose we have already assigned to some agents of $N_1$ sets of vertices from $V_1$, each set of value at least $\frac{1}{4}$ of the corresponding maximin share.
Let $N_1' \subseteq N_1$ be the set of agents that did not receive any set.
Let $(v_s,\ldots,v_{m_1})$ be the suffix of $o_v$ consisting of vertices not assigned yet.
Let $p$ be the least $t\geq s$ such that there is an agent $k \in N_1'$ for whom the value of the set $L=\{v_s\ldots,v_t\}$ is at least $\frac{1}{4}\mms_k$.

Notice that the value of the set $L$ for any agent $j \in N_1'$, including agent $k$, is 
\begin{equation}\label{multi1}
u_j(L)<\frac{1}{2}\mms_j.
\end{equation}
Indeed, by the definition of $p$, we have $u_j(L \setminus \{ v_p\})<\frac{1}{4}\mms_j$ and $u_j(v_p)<\frac{1}{4}\mms_j$.

We assign the set $L$ to agent $k$.
If all agents have their assigned sets, we terminate, and otherwise we continue the process with the subsequence $(v_{p+1},\ldots,v_{m_1})$.

Let us verify that the procedure described above assigns a set of vertices of a total value at least $\frac{1}{4}\mms_i$ to {\sl each} agent $i\in N_1$.
Suppose this is not the case, and only $q<n_1\leq n$ agents in $N_1$ received such sets.
Let $j\in N_1$ be an agent who did not receive a set.
Then the total value, for agent $j$, of the vertices of $V_1$ that were not assigned to any agent is less than $\frac{1}{4}\mms_j$. 
On the other hand, by the inequalities (\ref{multi1}) and (\ref{multi0}), this value is greater than $u_j(V_1)-\frac{q}{2}\mms_j\geq\frac{n-q}{2}\mms_j\geq\frac{1}{2}\mms_j$, a contradiction. Thus, we conclude that each agent $i\in N_1$ is assigned a subset of $V_1$ with total value at least $\frac{1}{4}\mms_i$ for that agent.
An analogous procedure assigns to every agent of $i\in N_2$ a subset of $V_2$ with total value at least $\frac{1}{4}\mms_i$ for that agent. 

Let $A_i$ denote the set of vertices assigned to agent $i\in N_1$ according to our procedure.
Clearly, the sets $A_i$ may not be bundles, i.e. they may not induce connected subgraphs of $G$.
Therefore, to ensure connectivity, we will now extend these sets with vertices from $G$ that have not yet been assigned to any agents. 
We can assume without loss of generality that  $n_1$ is the agent in $N_1$ who receives her set last. By the inequality (\ref{multi1}), the total value for agent $n_1$ of all sets assigned to agents in $N_1$ is less than $\frac{n_1}{2}\mms_{n_1}$. 
Thus, the total value for this agent of the unassigned vertices of $V_1$ is greater than $u_{n_1}(V_1)-\frac{n_1}{2}\mms_{n_1}\geq \frac{n-n_1}{2}\mms_{n_1}=\frac{n_2}{2}\mms_{n_1}$. Since the value of each vertex for this agent is smaller than $\frac{1}{4}\mms_{n_1}$, the number of unassigned vertices of $V_1$ is larger than $n_2$. We similarly show that the number of unassigned vertices of $V_2$ is larger than $n_1$. 

For each agent $i\in N_1$ (respectively, $i\in N_2$) we select a different unassigned vertex from $V_2$ (respectively, from $V_1$) and add it to the set $A_i$. The resulting set, say $A_i'$, is a (connected) bundle, which is assigned to agent $i$. There might still be some vertices which remain unassigned. However, for each $i \in N$, as $A_i'$ intersects both $V_1$ and $V_2$, every unassigned vertex has a neighbor in $A_i'$. Thus these vertices can be distributed arbitrarily. This way we obtain a $\frac{1}{4}$-mms-allocation for the graph $G$. \end{proof}

\section{Split graphs}
Now let us focus on split graphs. We prove Theorem~\ref{thm:splitintro}, which we again restate below.

\begin{theorem}
For every connected split graph and every set of agents of at most $2^k$ types, there exists a $\frac{3}{7\cdot 2^k - 3}$-mms-allocation.
\end{theorem}
\begin{proof}
Let $G$ be a connected split graph and let $I$ be the maximal independent set in $G$. 
Consider a set of $n$ agents with utility functions $u_1, \ldots, u_n$, ordered in such a way that each of the agent types is represented among the first $2^k$ agents (so that $u_i$ is the utility function of the $i$-th type of agents).

By Lemmas~\ref{prop_old} and~\ref{obs:heavy}, without loss of generality we may assume that all agents are $\frac{3}{7\cdot 2^k - 3}$-mms-bounded.

The general outline of our proof is the following. For each agent type $i$, we define $S_i$ to be an mms-partition for agents of this type. Then, for each vertex $v$ in $I$, we remove $v$ from bundles in partitions $S_i$ corresponding to all types of agents except one, effectively assigning $v$ to a single type of agent; we also make sure that each bundle retains a constant (depending on $k$) fraction of its original utility. This way we transform our sequence of $(G,n)$-partitions into a sequence of $(G,n)$-packings. Finally, we reduce the problem to mms-allocation on a complete graph.

We start by proving the following claim.

\begin{claim}
\label{claim_leaf_assignment}
There is a sequence of $2^k$ $(G,n)$-packings $S_1, S_2, \ldots, S_{2^k}$ such that
\begin{itemize}
\item[(i)] for every $i$ and every bundle $P$ in $S_i$, we have $u_i(P)\geq \frac{4}{7\cdot 2^k - 3} \mms^{(n)}(G, u_i)$,
\item[(ii)] for every vertex $v\in I$, there is exactly one $i$ such that $v$ is contained in a bundle in $S_i$.
\end{itemize}
\end{claim}

We will prove Claim~\ref{claim_leaf_assignment} using induction. For this reason, we reformulate its statement in the following way. 

\begin{claim}
\label{claim_leaf_ass_inductive}
For every choice of nonnegative utility functions $u_1, u_2, \ldots, u_{2^\ell}$, there exists a sequence of $2^\ell$ $(G,n)$-packings $( S_1, S_2, \ldots, S_{2^\ell} )$ such that 
\begin{itemize}
\item[(i)] for every $i$ and every bundle $P$ in $S_i$, we have $u_i(P)\geq \beta_{k,\ell} \mms^{(n)}(G, u_i)$,
\item[(ii)] for every vertex $v\in I$, there is exactly one $i\in \lbrace 1, 2, \ldots, 2^\ell\rbrace$ such that $v$ is contained in a bundle in $S_i$,
\end{itemize}
where $\beta_{k,\ell}$ is defined so that $\beta_{k, 0}=1$ and $\beta_{k,\ell+1}=\frac{1}{2}\left(\beta_{k, \ell} - \frac{3}{7\cdot 2^k - 3}\right)$.
\end{claim}
\begin{proof}
We will prove Claim~\ref{claim_leaf_ass_inductive} by induction on $\ell$. 
Note that for $\ell=0$ the claim follows directly from the definition of $mms_1$, i.e. we can take $S_1$ to be an mms-partition of $G$ for the agent with utility function $u_1$.

Now suppose that the claim is true for some $\ell\geq 0$ and consider a sequence of utility functions $u_1, \ldots, u_{2^{\ell+1}}$. Let $( R_1, R_2, \ldots, R_{2^\ell} )$ be the sequence of $2^\ell$ $(G,n)$-packings obtained by the induction hypothesis for utility functions $u_1,  \ldots, u_{2^\ell}$ and $( R_{2^\ell+1}, R_{2^\ell+2}, \ldots, R_{2^{\ell+1}} )$ be the sequence of $2^\ell$ $(G,n)$-packings obtained by the induction hypothesis for utility functions $u_{2^\ell+1}, u_{2^\ell+2}, \ldots, u_{2^{\ell+1}}$. Note that the concatenation of those sequences satisfies the following conditions 
\begin{itemize}
\item[(i')] for every $i\in \lbrace 1, 2, \ldots, 2^{\ell+1}\rbrace$ and every bundle $P$ in $R_i$, we have $u_i(P)\geq \beta_{k,\ell} \mms^{(n)}(G, u_i)$,
\item[(ii')] for every vertex $v\in I$, there are exactly two indices $i\in \lbrace 1, 2, \ldots, 2^{\ell+1}\rbrace$ such that $v$ is contained in a bundle in $R_i$.
\end{itemize}

Now we will assign each vertex from $I$ to one of the two bundles it is contained in, and remove it from the other bundle. We assume that initially no vertex in $I$ is assigned to any bundle. We decide on those assignments in a specific order, defined by the following procedure.

\begin{enumerate}[noitemsep]
\item Define a sequence of packings $R^\prime_1, R^\prime_2, \ldots, R^\prime_{2^{\ell+1}}$ to be identical to $R_1, R_2, \ldots, R_{2^{\ell+1}}$.
\item Pick a bundle $P$ from any packing $R^\prime_i$ such that $P$ contains an unassigned vertex from $I$.
\item Let $v$ be an unassigned vertex from $P \cap I$ that maximizes $u_i(v)$.
\item Assign $v$ to $P$ and, if $v$ is in another bundle $P'\in R^\prime_j$, remove $v$ from $P'$ (thereby altering $R^\prime_j$).
\item Go back to step 2 and pick $P'$ in that step, if $P'$ was defined in step 4 and contains an unassigned vertex; otherwise pick any bundle in step 2 satisfying the condition in that step.
\end{enumerate}
Note that after the procedure, we retain a sequence of $2^{\ell+1}$ $(G,n)$-packings; let us denote them by $S_1, S_2, \ldots, S_{2^{\ell+1}}$. Note that each set in each $S_i$ is indeed a bundle (i.e. induces a connected subgraph of $G$), because it was obtained from a bundle of $R_i$ by removing a number of vertices from $I$. Moreover, it clearly satisfies the condition that for every vertex $v\in I$, there is exactly one $i\in \lbrace 1, 2, \ldots, 2^{\ell+1}\rbrace$ such that $v$ is contained in a bundle in $S_i$ (i.e. the condition (ii) of Claim~\ref{claim_leaf_ass_inductive} for $\ell+1$).

Now, for a bundle $P$ in $S_i$, we will find a lower bound on $u_i(P)$. Let $Q$ be the bundle from $R_i$ that contains $P$, and let $V_Q=(v_1, v_2, \ldots)$ be a sequence of vertices of $Q$, ordered in descending order of value with respect to $u_i$, i.e. $u_i(v_1)\geq u_i(v_2)\geq \ldots$. Moreover, let $V_P=(w_1, w_2, \ldots)$ be a sequence of vertices of $P$, ordered in descending order of value with respect to $u_i$. To avoid tedious case analysis for $j$ greater than the size of $P$, we will assume that, for such $j$, $v_j$ is a phantom vertex of utility $0$; we will also assume that vertices admitting the same value of $u_i$ are in the same order in which they have been considered in the procedure, and vertices outside $I$ are in the same order in both $V_P$ and $V_Q$. 

We will argue that for every $j$, $u_i(w_j)\geq u_i(v_{2j})$. Indeed consider the set of vertices $X=\lbrace v_1, v_2, \ldots v_{2j} \rbrace$. If $X$ contains an even number of vertices from $I$, then each time a vertex from $X$ was removed from $Q$, some vertex from $X\cap I$ was assigned to $P$ immediately afterwards. In the other case take $m<j$ such that $X$ contains $2m+1$ vertices from $I$ and note that immediately after each of the first $m$ vertices from $X\cap I$ have been removed from $Q$, some vertex from $X\cap I$ was assigned to $P$, and $P$ contains additional $2j-2m-1$ vertices from $X\setminus I$. Therefore, in both cases $P$ contains at least $j$ vertices from $X$, which proves that $u_i(w_j)\geq u_i(v_{2j})$ as desired.

Now we are ready to bound $u_i(P)$. By the claim proved in the preceding paragraph and the inequality $u_i(v_{2j}) \geq u_i(v_{2j+1})$, we observe that 
\begin{align*}
u_i(P) \geq  \sum_j u_i(v_{2j}) &\geq  \frac{1}{2}\left( \sum_j u_i(v_{2j}) + \sum_j u_i(v_{2j+1}) \right) \\
&= \frac{1}{2}\left(u_i(Q)-u_i(v_1)\right).
\end{align*}

Recall that (without loss of generality) we assumed that agents are $\frac{3}{7\cdot 2^k - 3}$-mms-bounded and, by (i'), $u_i(Q)\geq \beta_{k,\ell} \mms^{(n)}(G, u_i)$, so
\[
u_i(P) \geq \frac{1}{2}\left(\beta_{k,\ell} - \frac{3}{7\cdot 2^k - 3} \right)\mms^{(n)}(G, u_i),
\]
which, by the definition of $\beta_{k, \ell + 1}$, coincides with (i) for $\ell+1$ and therefore, by induction, completes the proof of Claim~\ref{claim_leaf_ass_inductive}.
\end{proof}

Note that applying Claim~\ref{claim_leaf_ass_inductive} with $\ell = k$ immediately yields Claim~\ref{claim_leaf_assignment}, provided that $\beta_{k, k} = \frac{4}{7 \cdot 2^k - 3}$, which follows from the following straightforward calculation.

By unwrapping the recurrence 
$$\beta_{k,\ell+1}=\frac{1}{2}\left(\beta_{k, \ell} - \frac{3}{7\cdot 2^k - 3}\right), \quad \beta_{k, 0}=1$$
we observe that
$$
\beta_{k, k} = \frac{1}{2^k} - \frac{3}{2\left(7\cdot 2^k - 3\right)}\left(\frac{1}{2^{k-1}} + \frac{1}{2^{k-2}} + \ldots + 1\right).
$$
Using the formula for the sum of a geometric series we obtain
$$
\beta_{k, k} = \frac{1}{2^k} - \frac{3}{2\left(7\cdot 2^k - 3\right)}\left(\frac{2^k-1}{2^{k-1}}\right),
$$
then by performing a multiplication of fractions we get
$$
\beta_{k, k} = \frac{1}{2^k} - \frac{3\cdot 2^k - 3}{2^k\left(7\cdot 2^k - 3\right)},
$$
by reduction to a common denominator it implies that
$$
\beta_{k, k} = \frac{7\cdot 2^k - 3 - 3\cdot 2^k + 3}{2^k\left(7\cdot 2^k - 3\right)},
$$
so we finally obtain
$$
\beta_{k, k} = \frac{4}{7\cdot 2^k - 3}.
$$

Now consider a sequence of $2^k$ $(G,n)$-packings $S_1, S_2, \ldots, S_{2^k}$ obtained from Claim~\ref{claim_leaf_assignment}. Let $K:=G \setminus I$. For each vertex $v$ from $I$, let $S_{i'}$ be the unique packing in which it appears, and pick a neighbor $n_v\in K$ of $v$ such that $v$ and $n_v$ belong to the same bundle of $S_{i'}$. 
For a vertex $w$ in $K$ define $N_w:=\lbrace v\in I: w=n_v\rbrace$. Now define a sequence of nonnegative utility functions $u^\prime_1, u^\prime_2, \ldots, u^\prime_n$ on vertices of $K$ such that $u^\prime_i(w)=u_i(w)+u_i(N_w)$. Note that for each $i$, $S_i$ induces the partition $S^\prime_i = \lbrace P_j \cap K: P_j \in S_i \rbrace$ of $K$. By the definition of $u^\prime_i$, for each $P_j\in S_i$ we have $u^\prime_i(P_j\cap K)=u_i(P_j)$, so by (i) from Claim~\ref{claim_leaf_assignment} we conclude that $\mms^{(n)}(K, u^\prime_i)\geq \frac{4}{7\cdot 2^k - 3} \mms^{(n)}({G}, u_i)$. By Theorem~\ref{thm:completegraph34}, it follows that there is a $\frac{3}{4}$-mms-allocation\footnote{Note that a $(\frac{3}{4} + \frac{3}{3836})$-mms-allocation given by Theorem~\ref{thm:completegraph34} is also a $\frac{3}{4}$-mms-allocation. We ignore the $\frac{3}{3836}$ term in calculations, as it makes the final constant only slightly better but more complicated.} $(A^\prime_1, A^\prime_2, \ldots, A^\prime_n)$ of $K$ (with respect to utility functions $u^\prime_i$), so for each $i$ we have $u^\prime_i(A'_i)\geq \frac{3}{4} \mms^{(n)}(K, u^\prime_i) \geq \frac{3}{7\cdot 2^k - 3} \mms^{(n)}({G}, u_i)$. 

Consider now the allocation $(A_1, A_2, \ldots, A_n)$ such that $A_i=A^\prime_i \cup \bigcup_{w\in A^\prime_i}N_w$. Clearly $u_i(A_i) = u^\prime_i(A^\prime_i)$, and by the previous inequality it follows that $u_i(A_i) \geq \frac{3}{7\cdot 2^k - 3} \mms^{(n)}({G}, u_i)$. Therefore, $(A_1, A_2, \ldots, A_n)$ is a $\frac{3}{7\cdot 2^k - 3}$-mms-allocation of $G$ (with respect to utility functions $u_i$), which completes the proof.
\end{proof}

The following corollary is an immediate consequence of Theorem~\ref{thm:splitintro}. 
\begin{corollary}
For every connected split graph and every set of agents of $p>1$ types, there exists a $\frac{3}{14p-17}$-mms-allocation.
\end{corollary}
\begin{proof}
We define $k$ to be the unique integer such that $2^{k-1}+1\leq p\leq 2^k$. By Theorem~\ref{thm:splitintro}, there exists a $\frac{3}{7\cdot 2^k - 3}$-mms-allocation. Since $\frac{3}{7\cdot 2^k - 3}\geq\frac{3}{14p-17}$, this allocation is a $\frac{3}{14p-17}$-mms-allocation.
\end{proof}

\section{Final remarks and open problems}
It is worth noting that all previous results on approximate allocations that guarantee agents a positive fraction of
their maximin shares concerned graphs of goods without large induced stars. The presence of such stars in a graph seems to be a significant obstacle in proving approximation results of this kind. In contrast, all graph classes considered in this paper contain arbitrarily large induced stars. Thus, we believe that our main contribution lies not in showing approximation bounds for some specific graphs classes, but rather in developing a new set of  tools and techniques that can be used beyond previously considered settings.

Since mms-allocation always exists when the graph of goods is a tree, it seems to be natural to study approximate mms-allocations for graphs of bounded treewidth. In a tree, every inner vertex is a cut-vertex, which significantly restricts the way how the tree can be partitioned into connected subgraphs. As a consequence, in maximin allocations, agents have ``low expectations'' and are therefore easier to satisfy.
For graphs of bounded treewidth, the situation is similar. Instead of cut-vertices, we have a hierarchical structure of ``small'' separators. Still, only a bounded number of bundles can cross any edge in a tree decomposition. However, the problem remains non-trivial. First, ``bounded number of bundles'' is much more general than ``one bundle.'' Second, the structure of the separator itself seems to play an important role. 
Thus, a promising orthogonal direction for further study is to consider, for example, chordal graphs, where all separators are cliques.

The results proved in this paper are existential and do not yield polynomial-time algorithms for finding allocations.
We point out that they can be algorithmized in the following limited way: given an input graph of goods and, for each agent, her utility function and a corresponding mms-partition, in polynomial time we can compute the desired approximate allocation.
The bottleneck in fully algorithmizing the problem, i.e., not requiring that an mms-partition is given in the input, is the notorious hardness of computing mms-partitions. It might be possible to circumvent this issue but, for the moment, we do not know how to do it for the classes of graphs that we consider in this paper. 

In this paper we proved that, for the class of block-cactus graphs and complete multipartite graphs, there exists a constant $\alpha>0$ such that for any connected graph of goods in this class and arbitrary agents there is an $\alpha$-mms-allocation. In the Introduction we posed  a question (Question~\ref{quest1}) whether such a constant exists for the class of all graphs. 
The answer to this question remains unknown, not only for arbitrary agents but also for agents with a fixed number $k>1$ of agent types. If all agents are of the same type, then an mms-allocation always exists for any connected graph of goods (see Lonc and Truszczynski~\cite{lt:20}).  This leads us to ask the following question.
\begin{question}\label{quest2}
Is there a constant $\alpha>0$ such that for any connected graph of goods and any set of agents of $2$ types, there exists an $\alpha$-mms-allocation?
\end{question}
Finally, we do not know much about upper bounds for the constant $\alpha$ appearing in Questions~\ref{quest1} and~\ref{quest2}, even for restricted classes of graphs.
For a class $\cG$, let $\alpha_{\cG}$ denote the largest possible value of  $\alpha$ that would guarantee the existence of an $\alpha$-mms allocation for any graph in $\cG$.
Feige et al.~\cite{fst:22} proved that if $\cG$ is the class of complete graphs, then $\alpha_{\cG} \leq \frac{39}{40}$. Lonc and Truszczynski~\cite{lt:20} showed that if $\cG$ is the class of cycles, then $\alpha_\cG \leq \frac{3}{4}$.  
These results imply upper bounds for $\alpha_{\cG}$  for the classes $\cG$ of graphs considered in this paper. Indeed, the classes of block graphs, complete multipartite graphs, and split graphs contain all complete graphs, and the class of cactus graphs contains all cycles.
It would be interesting to narrow the gaps in the bounds on $\alpha_{\cG}$ (if possible) in these cases.

\newpage

\bibliography{mms-biblio}

\begin{thebibliography}{30}
\providecommand{\natexlab}[1]{#1}
\providecommand{\url}[1]{\texttt{#1}}
\expandafter\ifx\csname urlstyle\endcsname\relax
  \providecommand{\doi}[1]{doi: #1}\else
  \providecommand{\doi}{doi: \begingroup \urlstyle{rm}\Url}\fi

\bibitem[Akrami and Garg(2024)]{ag:24}
H.~Akrami and J.~Garg.
\newblock Breaking the 3/4 barrier for approximate maximin share.
\newblock In D.~P. Woodruff, editor, \emph{Proceedings of the 2024 Annual
  ACM-SIAM Symposium on Discrete Algorithms, {SODA} 2024}, pages 74--91. SIAM,
  2024.

\bibitem[Akrami et~al.(2023)Akrami, Garg, Sharma, and Taki]{agst:23}
H.~Akrami, J.~Garg, E.~Sharma, and S.~Taki.
\newblock Simplification and improvement of mms approximation.
\newblock In E.~Elkind, editor, \emph{Proceedings of the Thirty-Second
  International Joint Conference on Artificial Intelligence, {IJCAI} 2023},
  pages 2485--2493. IJCAI.org, 2023.

\bibitem[Amanatidis et~al.(2017)Amanatidis, Markakis, Nikzad, and
  Saberi]{amns:17}
G.~Amanatidis, E.~Markakis, A.~Nikzad, and A.~Saberi.
\newblock Approximation algorithms for computing maximin share allocations.
\newblock \emph{ACM Transactions on Algoritms (TALG)}, 13\penalty0
  (4):\penalty0 52:1--52:28, 2017.

\bibitem[Barman and Krishnamurthy(2020)]{bm:20}
S.~Barman and S.~K. Krishnamurthy.
\newblock Approximation algorithms for maximin fair division.
\newblock \emph{ACM Transactions on Economics and Computation}, 8\penalty0
  (1):\penalty0 1--28, 2020.

\bibitem[Bei et~al.(2022)Bei, Igarashi, Lu, and Suksompong]{bils:22}
X.~Bei, A.~Igarashi, X.~Lu, and W.~Suksompong.
\newblock The price of connectivity in fair division.
\newblock \emph{SIAM Journal on Discrete Mathematics}, 36\penalty0
  (2):\penalty0 1156--1186, 2022.

\bibitem[Bil{\`{o}} et~al.(2022)Bil{\`{o}}, Caragiannis, Flammini, Igarashi,
  Monaco, Peters, Vinci, and Zwicker]{Bilo:19}
V.~Bil{\`{o}}, I.~Caragiannis, M.~Flammini, A.~Igarashi, G.~Monaco, D.~Peters,
  C.~Vinci, and W.~S. Zwicker.
\newblock Almost envy-free allocations with connected bundles.
\newblock \emph{Games and Economic Behavior}, 131:\penalty0 197--221, 2022.

\bibitem[Biswas et~al.(2023)Biswas, Payan, Sengupta, and Viswanathan]{bpsv:23}
A.~Biswas, J.~Payan, R.~Sengupta, and V.~Viswanathan.
\newblock The theory of fair allocation under structured set constraints.
\newblock In A.~Mukherjee, J.~Kulshrestha, A.~Chakraborty, and S.~Kumar,
  editors, \emph{Ethics in Artificial Intelligence: Bias, Fairness and Beyond},
  pages 115--129. Springer, 2023.

\bibitem[Bondy and Murty(2008)]{bm:08}
J.~Bondy and U.~Murty.
\newblock \emph{Graph Theory}.
\newblock Springer, 2008.

\bibitem[Bouveret et~al.(2017)Bouveret, Cechl{\'{a}}rov{\'{a}}, Elkind,
  Igarashi, and Peters]{bceip:17}
S.~Bouveret, K.~Cechl{\'{a}}rov{\'{a}}, E.~Elkind, A.~Igarashi, and D.~Peters.
\newblock Fair division of a graph.
\newblock In C.~Sierra, editor, \emph{Proceedings of the 26th International
  Joint Conference on Artificial Intelligence, {IJCAI} 2017}, pages 135--141.
  ijcai.org, 2017.

\bibitem[Bouveret et~al.(2019)Bouveret, Cechl{\'{a}}rov{\'{a}}, and
  Lesca]{bcl:19}
S.~Bouveret, K.~Cechl{\'{a}}rov{\'{a}}, and J.~Lesca.
\newblock Chore division on a graph.
\newblock \emph{Autonomous Agents and Multi-Agent Systems}, 33\penalty0
  (5):\penalty0 540--563, 2019.

\bibitem[Brams and Taylor(1996)]{bt:96}
S.~J. Brams and A.~D. Taylor.
\newblock \emph{Fair Division: From Cake-Cutting to Dispute Resolution}.
\newblock Cambridge University Press, 1996.

\bibitem[Budish(2011)]{b:11}
E.~Budish.
\newblock The combinatorial assignment problem: Approximate competitive
  equilibrium from equal incomes.
\newblock \emph{Journal of Political Economy}, 119(6):\penalty0 1061--1103,
  2011.

\bibitem[Deligkas et~al.(2021)Deligkas, Eiben, Ganian, Hamm, and
  Ordyniak]{deli:21}
A.~Deligkas, E.~Eiben, R.~Ganian, T.~Hamm, and S.~Ordyniak.
\newblock The parameterized complexity of connected fair division.
\newblock In Z.-H. Zhou, editor, \emph{Proceedings of the 30th International
  Joint Conference on Artificial Intelligence, {IJCAI} 2021}, pages 139--145.
  ijcai.org, 2021.

\bibitem[Feige et~al.(2022)Feige, Sapir, and Tauber]{fst:22}
U.~Feige, A.~Sapir, and L.~Tauber.
\newblock A tight negative example for mms fair allocations.
\newblock In M.~Feldman, H.~Fu, and I.~Talgam-Cohen, editors, \emph{Web and
  Internet Economics, 17th International Conference, WINE 2021, Potsdam,
  Germany, December 14–17, 2021, Proceedings}, pages 355--372. Springer,
  2022.

\bibitem[Garg and Taki(2021)]{gt:19}
J.~Garg and S.~Taki.
\newblock An improved approximation algorithm for maximin shares.
\newblock \emph{Artificial Intelligence}, 300\penalty0 (103547), 2021.
\newblock URL \url{https://doi.org/10.1016/j.artint.2021.103547}.

\bibitem[Garg et~al.(2019)Garg, McGlaughlin, and Taki]{gmt:19}
J.~Garg, P.~McGlaughlin, and S.~Taki.
\newblock Approximating maximin share allocations.
\newblock In \emph{Proceedings of the 2nd Symposium on Simplicity in Algorithms
  (SOSA), volume 69}, pages 20:1--20:11. Schloss Dagstuhl - Leibniz-Zentrum fur
  Informatik, 2019.

\bibitem[Ghodsi et~al.(2018)Ghodsi, Hajiaghayi, Seddighin, Seddighin, and
  Yami]{ghssy:18}
M.~Ghodsi, M.~Hajiaghayi, M.~Seddighin, S.~Seddighin, and H.~Yami.
\newblock Fair allocation of indivisible goods: Improvements and
  generalizations.
\newblock In \emph{Proceedings of the 2018 ACM Conference on Economics and
  Computation, EC-2018}, pages 539--556, New York, NY, USA, 2018. ACM.

\bibitem[Greco and Scarcello(2020)]{gs:19}
G.~Greco and F.~Scarcello.
\newblock The complexity of computing maximin share allocations on graphs.
\newblock In \emph{The 34th {AAAI} Conference on Artificial Intelligence,
  {AAAI} 2020}, pages 2006--2013. {AAAI} Press, 2020.

\bibitem[Greco and Scarcello(2024)]{gs:24}
G.~Greco and F.~Scarcello.
\newblock Maxileximin envy allocations and connected goods.
\newblock In M.~Wooldridge, J.~Dy, and S.~Natarajan, editors, \emph{Proceedings
  of the 38th {AAAI} Conference on Artificial Intelligence}, pages 9713--9721.
  AAAI Press, 2024.

\bibitem[Hummel and Igarashi(2024)]{hi:24}
H.~Hummel and A.~Igarashi.
\newblock Keeping the harmony between neighbors: Local fairness in graph fair
  division.
\newblock In M.~Dastani, J.~S. Sichman, N.~Alechina, and V.~Dignum, editors,
  \emph{Proceedings of the 23rd International Conference on Autonomous Agents
  and Multiagent Systems}, pages 852--860. ACM, 2024.

\bibitem[Igarashi and Peters(2019)]{ip:19}
A.~Igarashi and D.~Peters.
\newblock Pareto-optimal allocation of indivisible goods with connectivity
  constraints.
\newblock In \emph{The 33rd {AAAI} Conference on Artificial Intelligence,
  {AAAI} 2019}, pages 2045--2052. {AAAI} Press, 2019.

\bibitem[King and Burton(1982)]{kb:82}
R.~King and S.~Burton.
\newblock Land fragmentation: Notes on a fundamental rural spatial problem.
\newblock \emph{Progress in Human Geography}, 6(4):\penalty0 475--94, 1982.

\bibitem[Lonc(2023)]{Lonc:23}
Z.~Lonc.
\newblock Approximating fair division on $d$-claw-free graphs.
\newblock In E.~Elkind, editor, \emph{Proceedings of the Thirty-Second
  International Joint Conference on Artificial Intelligence, {IJCAI} 2023},
  pages 2826--2834. IJCAI.org, 2023.

\bibitem[Lonc and Truszczynski(2020)]{lt:20}
Z.~Lonc and M.~Truszczynski.
\newblock {Maximin Share Allocations on Cycles}.
\newblock \emph{Journal of Artificial Intelligence Research}, 69:\penalty0
  613--655, 2020.

\bibitem[Procaccia(2016)]{p:16}
A.~Procaccia.
\newblock Cake-cutting algorithms.
\newblock In F.~Brandt, V.~Conitzer, U.~Endriss, J.~Lang, and A.~D. Procaccia,
  editors, \emph{Handbook of Computational Social Choice}, pages 311--329.
  Cambridge University Press, 2016.

\bibitem[Procaccia and Wang(2014)]{pw:14}
A.~D. Procaccia and J.~Wang.
\newblock Fair enough: guaranteeing approximate maximin shares.
\newblock In M.~Babaioff, V.~Conitzer, and D.~Easley, editors,
  \emph{Proceedings of the {ACM} Conference on Economics and Computation,
  {EC-2014}}, pages 675--692. {ACM}, 2014.

\bibitem[Steinhaus(1948)]{s:48}
H.~Steinhaus.
\newblock The problem of fair division.
\newblock \emph{Econometrica}, 16:\penalty0 101--104, 1948.

\bibitem[Suksompong(2019)]{s:19}
W.~Suksompong.
\newblock Fairly allocating contiguous blocks of indivisible items.
\newblock \emph{Discrete Applied Mathematics}, 260:\penalty0 227--236, 2019.

\bibitem[Suksompong(2021)]{s:21}
W.~Suksompong.
\newblock Constraints in fair division.
\newblock \emph{ACM SIGecom Exchanges}, 19\penalty0 (2):\penalty0 46--61, 2021.

\bibitem[Xiao et~al.(2023)Xiao, Qiu, and Huang]{xqh:23}
M.~Xiao, G.~Qiu, and S.~Huang.
\newblock Mms allocations of chores with connectivity constraints: New methods
  and new results.
\newblock In N.~Agmon, B.~An, A.~Ricci, and W.~Yeoh, editors, \emph{AAMAS '23:
  Proceedings of the 2023 International Conference on Autonomous Agents and
  Multiagent Systems}, pages 2886--2888. ACM, 2023.

\end{thebibliography}

\end{document}